\documentclass[letterpaper, 11pt]{article}

\usepackage[margin=1.0in]{geometry}
\usepackage{setspace} 
	\onehalfspace

\usepackage{amsmath} \allowdisplaybreaks
\usepackage{amssymb}
\usepackage{amsthm} %
\usepackage{mathtools}

\theoremstyle{definition}
\newtheorem{definition}{Definition}

\usepackage{pgf}

\theoremstyle{plain} 
\newtheorem{theorem}{Theorem}
\newtheorem{lemma}{Lemma}

\usepackage{multirow}
\usepackage{array}
\usepackage{caption}
\usepackage{lscape}
\usepackage{subcaption}

\usepackage{natbib}
\usepackage{bibentry}

\usepackage{booktabs}
\usepackage{graphicx}
\usepackage{appendix}
\usepackage{enumerate}

\usepackage{url}
\usepackage[bookmarks=false,hidelinks]{hyperref}

\usepackage{authblk}

\usepackage{tikz}
\usepackage{pgf}

\usepackage{pgfplots}
\usetikzlibrary{positioning}
\usetikzlibrary{arrows,shapes}
\usetikzlibrary{arrows.meta}
\usetikzlibrary{quotes}
\usetikzlibrary{calc}

\newcommand{\email}[1]{{\href{mailto:#1}{\nolinkurl{#1}}}}

\usepackage{bm}

\usepackage{xcolor}

\def\Pc{\mathcal{P}}

\def\Rb{\mathbb{R}}

\newcommand{\du}{\mathop{}\!\mathrm{d}u}

\renewcommand{\bar}[1]{\mkern 1.5mu\overline{\mkern-1.5mu#1\mkern-1.5mu}\mkern 1.5mu}
\renewcommand{\hat}{\widehat}

\renewcommand{\epsilon}{\varepsilon}

\title{Sensitivity of Wardrop Equilibria: Revisited}
\author{Mahdi Takalloo}
\author{Changhyun Kwon\thanks{Corresponding author: \texttt{chkwon@usf.edu}}}
\affil{Department of Industrial and Management Systems Engineering, University of South Florida}
\date{February 5, 2020}

\begin{document}
\maketitle

\begin{abstract}
For \emph{single-commodity} networks, the increase of the price of anarchy is bounded by a factor of $(1+\epsilon)^p$ from above, when the travel demand is increased by a factor of $1+\epsilon$ and the latency functions are polynomials of degree at most $p$. We show that the same upper bound holds for \emph{multi-commodity} networks and provide a lower bound as well.
\\[1em]
\end{abstract}

\section{Introduction and Notation}

We study Wardrop's traffic equilibria \citep{wardrop1952some} and how the price of anarchy changes with demand increases.
Wardrop's traffic equilibria is an example of nonatomic congestion games.
Nonatomic games \citep{schmeidler1973equilibrium} involve a continuum of players and congestion games \citep{rosenthal1973class} are a class of noncooperative Nash games where the utility of each player is a function of the number of total players who choose the same or overlapping strategies. 
The price of anarchy measures the inefficiency of equilibria \citep{koutsoupias1999worst,papadimitriou2001algorithms}
by comparing the worst-case social cost of equilibria to the social cost of the system optimal solution. 
The price of anarchy for nonatomic congestion games have been well studied in the literature \citep{roughgarden2002bad, roughgarden2005selfish, correa2008geometric}.

For \emph{single-commodity} networks, \citet{englert2010sensitivity} have provided the upper bound on the change of the price of anarchy when the demand increases.
A commodity is the travel demand for an origin-destination (O-D) pair and we assume that there is only one commodity for each O-D pair. 
In this paper, we show that the same upper bound is also valid for \emph{multi-commodity} networks.

We also provide a lower bound on the change of the price of anarchy when the demand increases for multi-commodity networks.
We utilize a classical sensitivity analysis approach, which has not been well recognized in the price of anarchy literature. 
By making connections between classical and modern approaches, 
we derive both upper and lower bounds on the changes of the price of anarchy, which were not straightforward using the methods available in the literature.

\section{Preliminaries}

For a given directed graph $G=(V,E)$, we consider non-decreasing latency functions $\ell_e : \Rb_{\geq 0} \mapsto \Rb_{\geq 0}$ for each edge $e\in E$.
For each commodity $i\in[k]=\{1,2,...,k\}$, the flow demand is $d_i$.
We let $\Pc_i$ denote the available paths for commodity $i$ and $\Pc = \cup_{i\in[k]} \Pc_i$.
Note that $\Pc_i \cap \Pc_j = \emptyset$ for any commodities $i \neq j$.
Let $(G,(d_i),\ell)$ denote an instance of Wardrop equilibrium problems.

A feasible path flow vector $f$ is feasible when $\sum_{P\in\Pc_i} f_P = d_i$ for all $i\in[k]$ and $f_p \geq 0$ for all $p\in\Pc$.
A path flow vector $f$ can also be written for each edge $e$, such that $f_e = \sum_{i\in[k]} \sum_{P\in\Pc_i:e\in P} f_P$.
The path latency is defined as $\ell_P(f) = \sum_{e\in P} \ell_e(f_e)$.
The total cost is defined as $C(f) = \sum_{P\in\Pc} \ell_P(f)f_P = \sum_{e\in E} \ell_e(f_e)f_e$. 
An optimal flow $f$ minimizes $C(f)$.

For each commodity $i \in [k]$, we define
\[
	\mu_i(f) = \min_{P\in\Pc_i} \ell_P(f_P).
\]
We consider the latency function for each edge $e\in E$ of the following form:
\[
	\ell_e(f_e) = \sum_{m=0}^p {b}_{em}  f_e ^m 
\]
for constants ${b}_{em}\geq 0$ for all $e\in E$ and $m=0,1,...,p$.

\begin{definition}
A feasible flow vector $f$ is at Wardrop equilibrium if
\[
	f_P > 0 \implies \ell_P(f) = \mu_i(f)
\]
for all $P\in\Pc_i$ and $i\in[k]$.
\end{definition}

It is well-known \citep{smith1979existence, dafermos1980traffic} that at any equilibrium flow $f$ for instance $(G, (d_i), \ell)$, we have
\begin{equation}\label{VIP}
	\sum_{e\in E} \ell_e(f_e) ( \bar{f}_e - f_e ) \geq 0 
\end{equation}
for all feasible $\bar{f}$ for instance $(G, (d_i), \ell)$.

Considering the demand increase, \citet{roughgarden2005selfish} derived the following bound for general non-decreasing latency functions:

\begin{theorem}[\citealp{roughgarden2005selfish}]
\label{thm:roughgarden}
Let $\ell_e(\cdot)$ be a non-decreasing function for all $e\in E$.
Let $C'_{\mathrm{opt}}$ be the cost of an optimal flow for instance $(G, ((1+\epsilon)d_i), \ell)$ and let $f$ be equilibrium flow for instance $(G,(d_i),\ell)$. 
\begin{equation} \label{vi}
	C(f) \leq \frac{1}{\epsilon} C'_{\mathrm{opt}}.
\end{equation}
\end{theorem}

If we assume polynomial functions for the latency, we can improve the bound in Theorem \ref{thm:roughgarden} using the following lemma:

\begin{lemma}[\citealp{christodoulou2011performance}] \label{lem:chris}
Let $\ell_e(\cdot)$ be a polynomial with nonnegative coefficients of degree $p$ for all $e\in E$. The inequality 
\[
  \ell_e(f_e) f'_e \leq \frac{p}{(p+1)^{1+1/p}} \ell_e(f_e) f_e +  \ell_e(f'_e) f'_e
\]
holds for any $f_e \geq 0 $ and $f'_e \geq 0$ for all $e\in E$.
\end{lemma}

The improved bound follows:

\begin{theorem} \label{thm:chris_bound}
Let $C'_{\mathrm{opt}}$ be the cost of an optimal flow for instance $(G, ((1+\epsilon)d_i), \ell)$ and let $f$ be equilibrium flow for instance $(G,(d_i),\ell)$. 
For both instances, suppose polynomial latency functions of degree at most $p$ with nonnegative coefficients.
We have
\begin{equation}
	C(f) \leq \Bigg[ (1+\epsilon) - \frac{p}{(p+1)^{1+1/p}} \Bigg]^{-1} C'_{\mathrm{opt}}.
\end{equation}
for all $\epsilon \geq 0 $.
\end{theorem}
\begin{proof}
For any feasible $\hat{f}$ for instance $(G, ((1+\epsilon)d_i), \ell)$, we let $\bar{f} = \frac{\hat{f}}{1+\epsilon}$ so that $\bar{f}$ is feasible for instance $(G, (d_i), \ell)$. From \eqref{VIP}, we have
\[
	\sum_{e\in E} \ell_e(f_e) f_e \leq \frac{1}{(1+\epsilon)} \sum_{e\in E} \ell_e(f_e) \hat{f}_e.
\]
The right-hand-side can be bounded using Lemma \ref{lem:chris} as follows:
\[
	\sum_{e\in E} \ell_e(f_e) \hat{f}_e \leq \frac{p}{(p+1)^{1+1/p}} \sum_{e\in E} \ell_e(f_e) f_e +  \sum_{e\in E} \ell_e(\hat{f}_e) \hat{f}_e
\]
Therefore, we have
\[
	\Bigg( 1 - \frac{\frac{p}{(p+1)^{1+1/p}}}{(1+\epsilon)} \Bigg) \sum_{e\in E} \ell_e(f_e) f_e \leq  \frac{1}{(1+\epsilon)} \sum_{e\in E} \ell_e(\hat{f}_e) \hat{f}_e
\]
and consequently
\[
	C(f) \leq \frac{ 1 } { (1+\epsilon) - \frac{p}{(p+1)^{1+1/p}} } C(\hat{f}).
\]
\end{proof}

Since $\frac{p}{(p+1)^{1+1/p}} < 1$ for all $p \geq 0$, Theorem \ref{thm:chris_bound} certainly improves Theorem \ref{thm:roughgarden} by considering a specific form of latency functions.
Note that we obtain Theorem \ref{thm:chris_bound} by comparing an equilibrium flow with any feasible flow. 
Using the following result, however, we will compare the performances of equilibrium flows and obtain a tighter bound.

\begin{theorem}[\citealp{dafermos1984sensitivity}] 
\label{thm:dafermos}
Let $\ell_e(\cdot)$ be a non-decreasing function for all $e\in E$.
Let $f$ and $f'$ be equilibrium flows for instances $(G,(d_i),\ell)$ and $(G, ((1+\epsilon)d_i), \ell)$, respectively.
Then the following inequality holds:
\begin{equation}
	\sum_{i\in [k]} ( \mu_i(f') - \mu_i(f) ) ( d'_i - d_i ) \geq 0
\end{equation}
where $d'_i = (1+\epsilon) d_i$. 
\end{theorem}

Theorem \ref{thm:dafermos} shows that the minimum path latency function $\mu_i(\cdot)$ exhibits the characteristics of monotone functions with respect to the travel demand changes. 
In Theorem \ref{thm:equ_bound}, we show that the ratio between the performances of equilibrium flows is bounded by $\frac{1}{1+\epsilon}$. 
See Figure \ref{fig:compare} for comparison.

\begin{figure} \centering
\resizebox{0.7\textwidth}{!}{\input{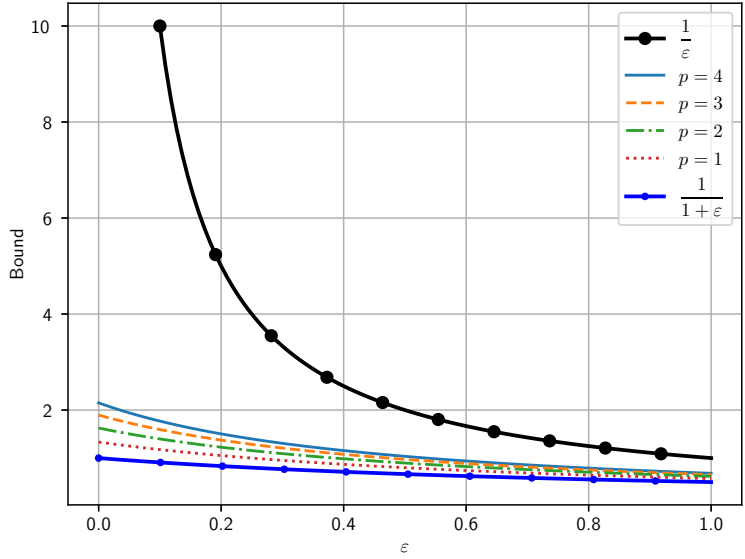}}
\caption{$\frac{1}{\epsilon}$ from Theorem \ref{thm:roughgarden}; $p=1,2,3,4$ from Theorem \ref{thm:chris_bound}; and $\frac{1}{1+\epsilon}$ from Theorem \ref{thm:equ_bound}}
\label{fig:compare}
\end{figure}

\section{Changes of the Price of Anarchy}

We consider the demand changes from $d_i$ to $(1+\epsilon)d_i$ for all commodity $i\in[k]$ for some $\epsilon\geq 0$.
We first compare the performances of system optimum flows in the two instances.

\begin{theorem} \label{thm:opt_bound}
Let $C_{\mathrm{opt}}$ and $C'_{\mathrm{opt}}$ be the cost of an optimal flow for instances $(G,(d_i),\ell)$ and $(G, ((1+\epsilon)d_i), \ell)$ with polynomial latency functions of degree at most $p$ with nonnegative coefficients, respectively.
Then we can show
\begin{align}
	(1+\epsilon) C_{\mathrm{opt}} &\leq C'_{\mathrm{opt}} \leq (1+\epsilon)^{p+1} C_{\mathrm{opt}}.
\end{align}
\end{theorem}

\begin{proof}
We let $f'^{*}$ be the optimal flow for instance $(G, ((1+\epsilon)d_i), \ell)$ and $\hat{f} = \frac{f'^*}{1+\epsilon}$.
Then,
\begin{align*}
	C'_{\mathrm{opt}} 
		& = \sum_{e\in E} \ell_e( (1+\epsilon) \hat{f}_e) (1+\epsilon) \hat{f}_e \\
    & = (1+\epsilon) \sum_{e\in E}\ell_e( (1+\epsilon)\hat{f}_e) \hat{f}_e \\
    & \geq (1+\epsilon) \sum_{e\in E}\ell_e(\hat{f}_e) \hat{f}_e \\                 
    & \geq (1+\epsilon) C_{\mathrm{opt}}.
\end{align*}
Also, we let $f^{*}$ be the optimal flow for instance $(G, (d_i), \ell)$ and then $(1+\epsilon) f^*$ is feasible to instance $(G, ((1+\epsilon)d_i), \ell)$.
We have
\begin{align*}
	C'_{\mathrm{opt}} 
		& \leq \sum_{e\in E} \ell_e( (1+\epsilon) f^*_e) (1+\epsilon) f^*_e \\
		& = \sum_{e\in E} \bigg( \sum_{m=0}^p {b}_{em} ( (1+\epsilon) f^*_e )^m \bigg) (1+\epsilon) f^*_e\\		
		& \leq \sum_{e\in E} \bigg(\sum_{m=0}^p {b}_{em} (1+\epsilon)^p (f^*_e)^m \bigg) (1+\epsilon) f^*_e\\				
    & = (1+\epsilon)^{p+1} \sum_{e\in E} \ell_e(f^*_e) f^*_e \\
    & = (1+\epsilon)^{p+1} C_{\mathrm{opt}}.
\end{align*}
\end{proof}

Next, we compare the performances of equilibrium flows in the two instances.
Although Theorem 3 of \citet{englert2010sensitivity} considers single-commodity networks and focuses on path latency, the same technique is valid for showing the following theorem for multi-commodity networks.
While the bound on the path latency does not hold in multi-commodity networks as noted by \citet{englert2010sensitivity}, it still provides a bound on the total cost.
Using Theorem \ref{thm:dafermos}, we can also provide a lower bound.

\begin{theorem} \label{thm:equ_bound}
Let $f$ and $f'$ be equilibrium flows for instances $(G,(d_i),\ell)$ and $(G, ((1+\epsilon)d_i), \ell)$ with polynomial latency functions of degree at most $p$ with nonnegative coefficients, respectively.
Then we can show
\begin{align}
	(1+\epsilon) C(f) &\leq C(f') \leq (1+\epsilon)^{p+1} C(f).
\end{align}
\end{theorem}

\begin{proof}
By Theorem \ref{thm:dafermos}, we have
\begin{align*}
	0 
		& \leq \sum_{i\in [k]} ( \mu_i(f') - \mu_i(f) ) ( (1+\epsilon) d_i - d_i ) \\
		& = \epsilon \sum_{i\in [k]} ( \mu_i(f') - \mu_i(f) ) d_i  \\
		& = \frac{\epsilon}{1+\epsilon} \sum_{i\in [k]} \mu_i(f') (1+\epsilon) d_i - \epsilon \sum_{i\in [k]} \mu_i(f) d_i \\		
		& = \frac{\epsilon}{1+\epsilon} \sum_{i\in [k]} \mu_i(f') d'_i - \epsilon \sum_{i\in [k]} \mu_i(f) d_i \\		
		& = \frac{\epsilon}{1+\epsilon} C(f') - \epsilon C(f),
\end{align*}
which leads to $(1+\epsilon) C(f) \leq C(f')$.

The upper bound, $C(f') \leq (1+\epsilon)^{p+1} C(f)$, is already proved by the proof of Theorem 3 in \citet{englert2010sensitivity}.
\end{proof}

In Theorems \ref{thm:opt_bound} and \ref{thm:equ_bound}, both lower and upper bounds are tight.
The lower bound happens when the latency functions are constant in all edges.
The upper bound happens when the latency function is a monomial of degree $p$ in a single-edge network with a single commodity.

For any $p\geq 0$, we have
\[
	\frac{1}{(1+\epsilon)} \leq \frac{ 1 } { (1+\epsilon) - \frac{p}{(p+1)^{1+1/p}} } < \frac{1}{\epsilon}.
\]
Therefore the bound in Theorem \ref{thm:equ_bound} is tighter than the bounds in Theorems \ref{thm:roughgarden} and \ref{thm:chris_bound}, as seen in Figure \ref{fig:compare}.

When the demand increases, from Theorems \ref{thm:opt_bound} and \ref{thm:equ_bound}, we can observe that the cost of both the optimal flow and the equilibrium flow increases at least by factor of $1+\epsilon$.
We obtain both lower and upper bounds on the change of price of anarchy as follows:

\begin{theorem} \label{thm2}
Let $\rho$ and $\rho'$ denote the Price of Anarchy (PoA) for instances $(G,(d_i),\ell)$ and $(G, ((1+\epsilon)d_i), \ell)$ with polynomial latency functions of degree at most $p$ with nonnegative coefficients, respectively.
Then $\frac{1}{(1+\epsilon)^{p}} \leq \frac{\rho'}{\rho} \leq (1+\epsilon)^p $.
\end{theorem}
\begin{proof}
We can show that
\[
\frac{\rho'}{\rho} 
= \frac{ C(f')/C'_{\textrm{opt}} }{ C(f)/C_{\textrm{opt}} } 
= \frac{C(f')}{C(f)} \cdot \frac{C_{\textrm{opt}}}{C'_{\textrm{opt}}} 
\leq (1+\epsilon)^{p+1} \cdot \frac{1}{1+\epsilon} 
= (1+\epsilon)^p
\]
where the inequality holds by Theorem \ref{thm:opt_bound}.
Similarly,
\[
\frac{\rho}{\rho'} 
= \frac{ C(f)/C_{\textrm{opt}} }{ C(f')/C'_{\textrm{opt}} } 
= \frac{C(f)}{C(f')} \cdot \frac{C'_{\textrm{opt}}}{C_{\textrm{opt}}} 
\leq \frac{1}{1+\epsilon} \cdot (1+\epsilon)^{p+1}
= (1+\epsilon)^p
\]

\end{proof}

The upper bound is identical to the result of \citet{englert2010sensitivity}, but holds for multi-commodity networks.
\citet{o2016mechanisms} study how the price of anarchy may decay as the demand increases. 
When the price of anarchy decreases, the lower bound in Theorem \ref{thm2} provides useful information.

\section{Examples and Insights}

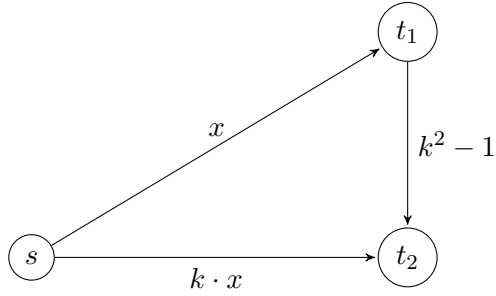
\begin{figure}\centering
	\begin{tikzpicture}[>=stealth',shorten >=1pt,auto,node distance=2.8cm]
		\node[shape=circle, draw] (s)  at (0,0) {$s$};
		\node[shape=circle, draw] (t2) at (5,0) {$t_2$};
		\node[shape=circle, draw] (t1) at (5,3) {$t_1$};		
		\path[->] (s) edge [above] node {$x$} (t1);
		\path[->] (t1) edge [right] node {$k^2-1$} (t2);
		\path[->] (s) edge [below] node {$k\cdot x$} (t2);
	\end{tikzpicture}
	\caption{An example \citep{englert2010sensitivity} with two commodities: from $s$ to $t_1$ and from $s$ to $t_2$}
	\label{fig:example}
\end{figure}

Consider an example in Figure \ref{fig:example}, originally considered in \citet{englert2010sensitivity}.
The edge latency function $\ell_e(x)$ is written on each edge as a function of edge flow $x$.
There are two commodities with demand $d_1=1$ and $d_2=k$ for constant $k\geq 1$. 
Suppose we increase the travel demand by factor $(1+\epsilon)$ for both commodities.
At the equilibrium, the path latency for $d_1$ increases from $1$ to $1+k\epsilon$ and the path latency for $d_2$ increases from $k^2$ to $k^2+k\epsilon$; by a multiplicative factor of $(1+k\epsilon)$ and $(1+\frac{\epsilon}{k})$, respectively.
For $d_1$, the multiplicative increase factor exceeds the bound $(1+\epsilon)^p$, while it is below the bound for $d_2$. 
Although the increase in the travel demand for commodities is uniform, the increase in the resulting path latency is not.

The total cost, however, is still bounded as indicated in Theorem \ref{thm:equ_bound}.
Before the increase, the total cost at the equilibrium is 
$C(f) = 1 + k^3$, while after the increase, it is $C(f') = (1+\epsilon)( 1+2k\epsilon + k^3)$.
Since $k\geq 1$, it is easy to show that the ratio $C(f')/C(f)$ is bounded below and above as follows:
\[
(1+\epsilon) \leq \frac{(1+\epsilon)( 1+2k\epsilon + k^3)}{1+k^3} \leq (1+\epsilon)^2.
\]
Note that when $k=1$, the ratio is equal to the upper bound. 
The lower bound becomes tight only when $\epsilon=0$. 
It is easy to see, however, that the lower bound in Theorem \ref{thm:equ_bound} becomes tight when all edge latency functions are constant functions.

We can obtain an insight from this example for the upper bound.
At the equilibrium $f$, the total cost is the sum of the path latency, weighted by the travel demand:
\[
	C(f) = \sum_{i\in[k]} \mu_i(f) d_i.
\]
Since the increase in $C(f)$ is bounded, an over-increase in the path latency for a certain commodity is alleviated by under-increases in the path latency for other commodities.
Those commodities with under-increases likely have more total travel demands than those with over-increases.
In the example in Figure \ref{fig:example}, $d_1=1$ and $d_2=k$ with $k\geq 1$; therefore the travel demand for the second commodity with under-increase in the path latency exceeds the travel demand for the first commodity with over-increase.

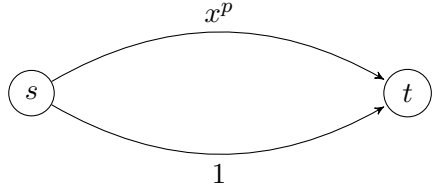
\begin{figure}\centering
	\begin{tikzpicture}[>=stealth',shorten >=1pt,auto,node distance=2.8cm]
		\node[shape=circle, draw] (s) at (0,0) {$s$};
		\node[shape=circle, draw] (t) at (5,0) {$t$};
		\path[->] (s) edge [bend left] node {$x^p$} (t);
		\path[->] (s) edge [bend right] node [below] {$1$} (t);
	\end{tikzpicture}
	\caption{Nonlinear Pigou Example \citep{roughgarden2005selfish}}
	\label{fig:pigou_example}
\end{figure}

\begin{figure} \centering
	\resizebox{0.5\textwidth}{!}{\input{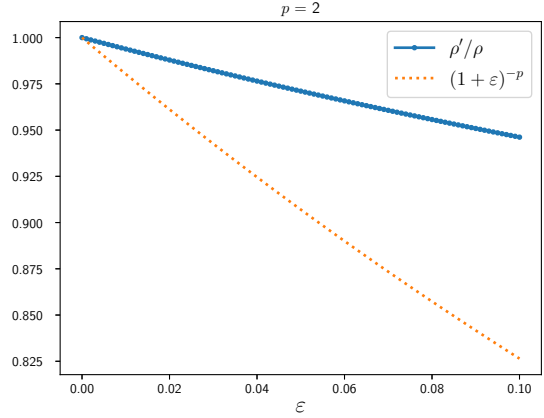}}
	\caption{Price of Anarchy decreases as $\epsilon$ increases in the nonlinear Pigou example}
	\label{fig:pigou}	
\end{figure}

We consider the nonlinear Pigou example in Figure \ref{fig:pigou_example}, which consists of a single commodity from $s$ to $t$ through two edges. 
In this example, the PoA decreases as $\epsilon$ increases.
When the travel demand is $1$, the PoA is at its upper bound, $\rho = \frac{(p+1)^{(1+1/p)}}{(p+1)^{(1+1/p)} - p}$.
If the travel demand increases to $1+\epsilon$, then the total cost at equilibrium is $1+\epsilon$ and the optimal total cost is $1+\epsilon - \frac{p}{1+p} \big(\frac{1}{1+p} \big)^{1/p}$.
Therefore, the PoA becomes
\[
	\rho' 
	= 
	\frac{1+\epsilon}{1+\epsilon - \frac{p}{1+p} \big(\frac{1}{1+p} \big)^{1/p}}
	=
	\bigg[ 1 - \frac{1}{1+\epsilon} \frac{p}{1+p} \Big(\frac{1}{1+p} \Big)^{1/p} \bigg]^{-1},
\]
which monotonically decreases as $\epsilon$ increases.
In Figure \ref{fig:pigou}, $\rho'/\rho$ is compared with the lower bound in Theorem \ref{thm2}.

\tikzset{shifted path/.style args={from #1 to #2 with label #3}{insert path={
let \p1=($(#1.east)-(#1.center)$),
\p2=($(#2.east)-(#2.center)$),\p3=($(#1.center)-(#2.center)$),
\n1={veclen(\x1,\y1)},\n2={veclen(\x2,\y2)},\n3={atan2(\y3,\x3)} in
(#1.{\n3+180+asin(\pgfkeysvalueof{/tikz/shifted path/dist}/\n1)}) edge[/tikz/shifted path/arrows,"#3"] (#2.{\n3-asin(\pgfkeysvalueof{/tikz/shifted path/dist}/\n2)})
}},back and forth/.style={/utils/exec=\pgfkeys{/tikz/shifted path/.cd,#1},
shifted path=from \pgfkeysvalueof{/tikz/shifted path/from} to \pgfkeysvalueof{/tikz/shifted path/to} with label \pgfkeysvalueof{/tikz/shifted path/label 1},
shifted path=from \pgfkeysvalueof{/tikz/shifted path/to} to \pgfkeysvalueof{/tikz/shifted path/from} with label \pgfkeysvalueof{/tikz/shifted path/label 2}},
shifted path/.cd,dist/.initial=2pt,arrows/.style={-stealth},from/.initial=1,to/.initial=2,label 1/.initial={},label 2/.initial={}}

\begin{figure}\centering
\resizebox{2in}{!}{
	\begin{tikzpicture}[>=stealth',shorten >=1pt,auto,node distance=2cm]
		\node[shape=circle, draw] (1)  {1};
		\node[shape=circle, draw] (3) [below of=1] {3};
		\node[shape=circle, draw] (4) [right of=3] {4};
		\node[shape=circle, draw] (5) [right of=4] {5};
		\node[shape=circle, draw] (6) [right of=5] {6};
		\node[shape=circle, draw] (2) [above of=6] {2};
		\node[shape=circle, draw] (9) [below of=5] {9};
		\node[shape=circle, draw] (8) [right of=9] {8};
		\node[shape=circle, draw] (7) [right of=8] {7};
		\node[shape=circle, draw] (18) [below of=7] {18};
		\node[shape=circle, draw] (16) [left of=18] {16};
		\node[shape=circle, draw] (10) [left of=16] {10};
		\node[shape=circle, draw] (11) [left of=10] {11};
		\node[shape=circle, draw] (12) [left of=11] {12};
		\node[shape=circle, draw] (17) [below of=16] {17};
		\node[shape=circle, draw] (19) [below of=17] {19};
		\node[shape=circle, draw] (15) [left of=19] {15};
		\node[shape=circle, draw] (14) [left of=15] {14};
		\node[shape=circle, draw] (23) [below of=14] {23};
		\node[shape=circle, draw] (22) [right of=23] {22};
		\node[shape=circle, draw] (24) [below of=23] {24};
		\node[shape=circle, draw] (13) [left of=24] {13};
		\node[shape=circle, draw] (21) [right of=24] {21};
		\node[shape=circle, draw] (20) [right of=21] {20};

   \draw[shifted path/arrows/.style={stealth-},
   back and forth/.list={
   	{from=1,to=2,label 1=,label 2=},
    {from=1,to=3,label 1=,label 2=},
    {from=2,to=6,label 1=,label 2=},
    {from=3,to=4,label 1=,label 2=},
    {from=4,to=5,label 1=,label 2=},
    {from=5,to=6,label 1=,label 2=},
    {from=3,to=12,label 1=,label 2=},
    {from=4,to=11,label 1=,label 2=},
    {from=5,to=9,label 1=,label 2=},
    {from=6,to=8,label 1=,label 2=},
    {from=9,to=8,label 1=,label 2=},
    {from=8,to=7,label 1=,label 2=},
    {from=12,to=11,label 1=,label 2=},
    {from=11,to=10,label 1=,label 2=},
    {from=9,to=10,label 1=,label 2=},
    {from=8,to=16,label 1=,label 2=},
    {from=7,to=18,label 1=,label 2=},
    {from=10,to=16,label 1=,label 2=},
    {from=16,to=18,label 1=,label 2=},
    {from=9,to=10,label 1=,label 2=},
    {from=12,to=13,label 1=,label 2=},
    {from=11,to=14,label 1=,label 2=},
    {from=10,to=15,label 1=,label 2=},
    {from=10,to=17,label 1=,label 2=},
    {from=16,to=17,label 1=,label 2=},
    {from=17,to=19,label 1=,label 2=},
    {from=14,to=15,label 1=,label 2=},
    {from=15,to=19,label 1=,label 2=},
    {from=14,to=23,label 1=,label 2=},
    {from=15,to=22,label 1=,label 2=},
    {from=19,to=20,label 1=,label 2=},
    {from=18,to=20,label 1=,label 2=},
    {from=23,to=22,label 1=,label 2=},
    {from=23,to=24,label 1=,label 2=},
    {from=22,to=21,label 1=,label 2=},
    {from=22,to=20,label 1=,label 2=},
    {from=13,to=24,label 1=,label 2=},
    {from=24,to=21,label 1=,label 2=},
    {from=21,to=20,label 1=,label 2=},
    }];

	\end{tikzpicture}
}
	\caption{Sioux Falls Network with 24 nodes, 76 edges, and 528 commodities}
	\label{fig:siouxfalls}
\end{figure}
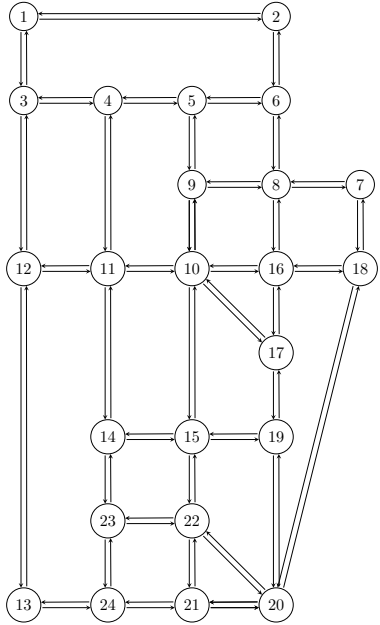

We test the Sioux Falls network for a single commodity following \citet{o2016mechanisms}. 
The Sioux Falls network, shown in Figure \ref{fig:siouxfalls}, has been popularly used in transportation research and the dataset is available online \citep{TNTP}.
The edge latency functions are polynomials of order $p=4$.
We choose a single commodity from node 20 to node 3 and set the initial travel demand as 1,000.
Then we increase the demand by 10\% each time by multiplying $(1+\epsilon)$ with $\epsilon=0.1$.
We repeat 50 times.
We compute both Wardrop equilibrium and system optimal flows.
The resulting price of anarchy (PoA) is presented in Figure \ref{fig:PoA_single}.
As the travel demand increases, PoA tends to increase and then decrease, but not monotonically.
There are several up-and-down points.
In each travel demand increase, the ratio between two PoA values, namely $\rho'/\rho$ as in Theorem \ref{thm2} is computed and presented in Figure \ref{fig:ratio_single}.
While the ratio is bounded between $(1+\epsilon)^{-4}=0.683$ and $(1+\epsilon)^4=1.464$ as given in Theorem \ref{thm2}, the ratios remain near 1.0 mostly.

\begin{figure} \centering
	\begin{subfigure}[b]{0.49\textwidth}\centering
	\resizebox{\textwidth}{!}{\input{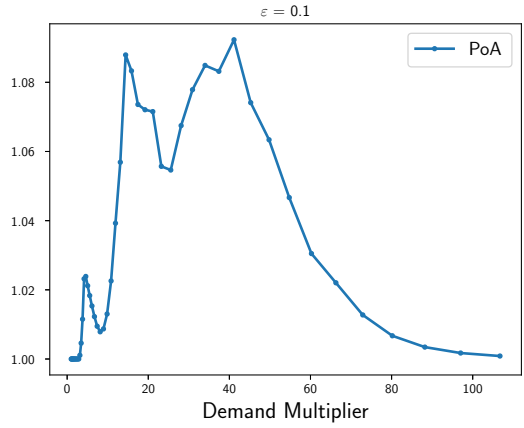}}
	\caption{}
	\label{fig:PoA_single}
	\end{subfigure}
	\begin{subfigure}[b]{0.49\textwidth}\centering
	\resizebox{\textwidth}{!}{\input{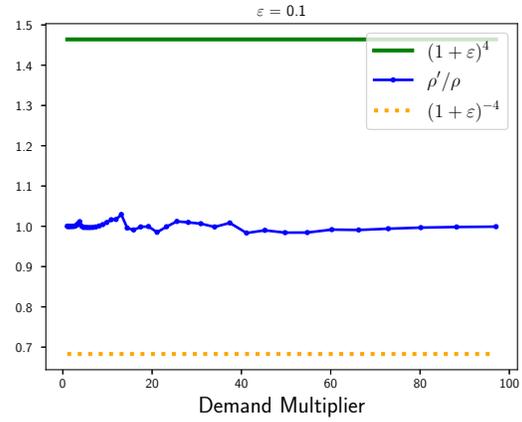}}
	\caption{}
	\label{fig:ratio_single}
	\end{subfigure}
	\caption{Price of Anarchy in Sioux Falls network with single commodity}
	\label{}	
\end{figure}

In the original dataset of the Sioux Falls network, there are 528 commodities with non-zero travel demand. 
In this time, we consider all 528 commodities.
We set the initial travel demand as 5\% of the original travel demand given in the dataset and then start increasing demands by factor of $(1+\epsilon)$ with $\epsilon=0.1$.
We repeat this process 40 times.
The results are presented in Figure \ref{fig:full}.
Similar observations can be made for both single-commodity and 528-commodity cases.

\begin{figure} \centering
	\begin{subfigure}[b]{0.49\textwidth}\centering
	\resizebox{\textwidth}{!}{\input{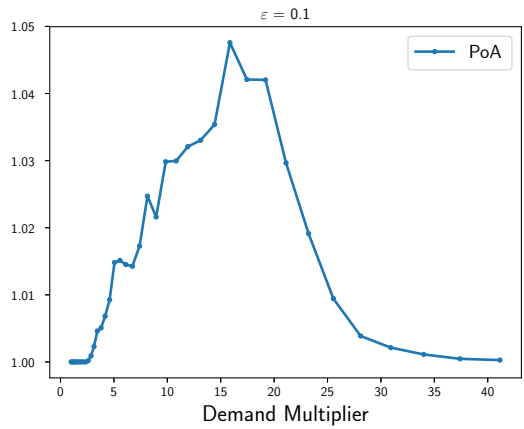}}
	\caption{}
	\label{fig:PoA_full}
	\end{subfigure}
	\begin{subfigure}[b]{0.49\textwidth}\centering
	\resizebox{\textwidth}{!}{\input{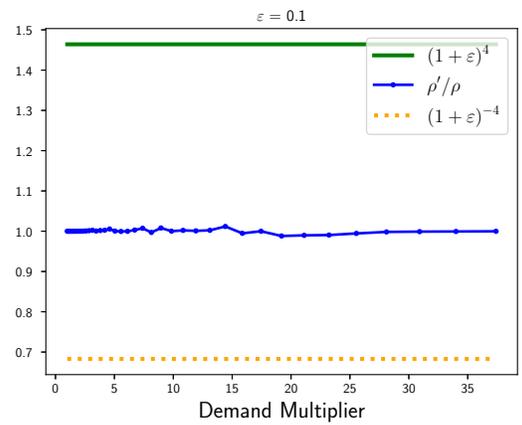}}
	\caption{}
	\label{fig:ratio_full}
	\end{subfigure}
	\caption{Price of Anarchy in Sioux Falls network with 528 commodities}
	\label{fig:full}
\end{figure}

To observe the effect of $\epsilon$ on the PoA ratio, we present the PoA ratio $\frac{\rho'}{\rho}$ for various $\epsilon$ values for single-commodity and 528-commodity Sioux Falls network in Figure \ref{fig:eps3}.
For the single OD pair Sioux Falls network, we set the initial demand to 3000 units.
For the full Sioux Falls network, we set the initial demand to the original travel demand.
As presented in Figure \ref{fig:eps1}, for the single commodity network, the PoA ratio is not monotone and it has a breakpoint.
The PoA ratio initially increases with $\epsilon$ with an increasing rate until the breakpoint, after which it decreases with a decreasing rate up to a local minimum point and then increases again.
For the full Sioux Falls network, shown in Figure \ref{fig:eps2}, the PoA ratio decreases with a non-increasing rate.

\begin{figure} \centering
	\begin{subfigure}[b]{0.49\textwidth}\centering
	\resizebox{\textwidth}{!}{\input{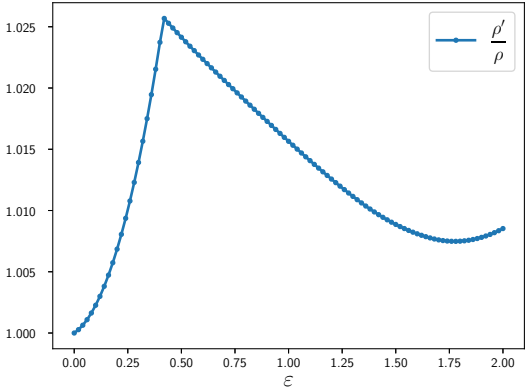}}
	\caption{Sioux Falls network with single commodity}
	\label{fig:eps1}
	\end{subfigure}
	\begin{subfigure}[b]{0.49\textwidth}\centering
	\resizebox{\textwidth}{!}{\input{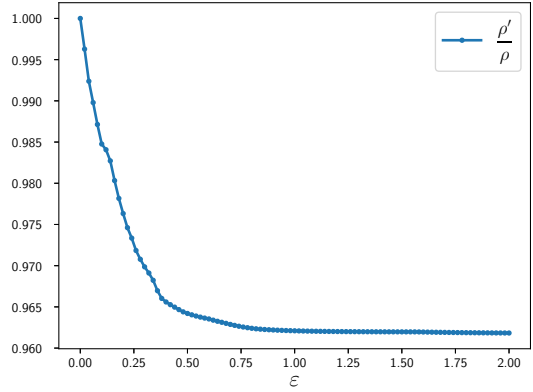}}
	\caption{Sioux Falls network with 528 commodities}
	\label{fig:eps2}
	\end{subfigure}
	\caption{PoA ratio for different $\epsilon$ values for Sioux Falls network}
	\label{fig:eps3}	
\end{figure}

\section{Discussion on the Non-tightness of the Bounds}

In the examples, we observe that the bounds in Theorem \ref{thm2} are not tight.
Indeed, it cannot be tight for large $\epsilon$.
Clearly, the PoA $\rho$ is bounded as follows \citep{roughgarden2005selfish}:
\[
	1 \leq \rho \leq \frac{(p+1)^{(1+1/p)}}{(p+1)^{(1+1/p)} - p}
\]
and so is $\rho'$.
Therefore, 
\[
	\frac{(p+1)^{(1+1/p)} - p}{(p+1)^{(1+1/p)}} \leq \frac{\rho'}{\rho} \leq \frac{(p+1)^{(1+1/p)}}{(p+1)^{(1+1/p)} - p} 
\]
also holds regardless $\epsilon$.
This indicates that the bounds in Theorem \ref{thm2} will never be tight for large $\epsilon$ such that
\begin{equation} \label{large_e}
\epsilon \geq \bigg[ \frac{(p+1)^{(1+1/p)}}{(p+1)^{(1+1/p)} - p} \bigg]^{\frac{1}{p}} - 1.
\end{equation}
This effective $\epsilon$ is plotted in Figure \ref{fig:large_e}.

\begin{figure} \centering
	\resizebox{0.5\textwidth}{!}{\input{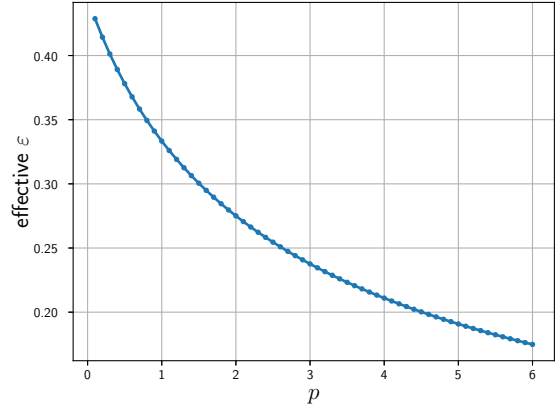}}
	\caption{The maximum $\epsilon$ effective for the bounds in Theorem \ref{thm2}, as given in \eqref{large_e}}
	\label{fig:large_e}	
\end{figure}

Figure \ref{fig:eps6} represents the PoA ratio, theoretical bound $(1+\epsilon)^p$ and the clear bound $\frac{(p+1)^{(1+1/p)}}{(p+1)^{(1+1/p)} - p}$ based on \citet{roughgarden2005selfish} for different $\epsilon$ values for single commodity and full Sioux Falls network ($p=4$). 
As  Figure \ref{fig:eps4} and \ref{fig:eps5} represent the PoA ratio is always less than the theoretical bounds. 
The two values, $(1+\epsilon)^p$ and $\frac{(p+1)^{(1+1/p)}}{(p+1)^{(1+1/p)} - p}$, become the same at about $\epsilon \approx 0.2110$.
When $\epsilon < 0.2110$, $(1+\epsilon)^p$ is smaller compare to
$\frac{(p+1)^{(1+1/p)}}{(p+1)^{(1+1/p)} - p}$; for larger $\epsilon$ values, it is greater than $\frac{(p+1)^{(1+1/p)}}{(p+1)^{(1+1/p)} - p}$.
Hence, the obtained bound $(1+\epsilon)^p$ can be tight only for small $\epsilon$ values.

\begin{figure} \centering
	\begin{subfigure}[b]{0.49\textwidth}\centering
	\resizebox{\textwidth}{!}{\input{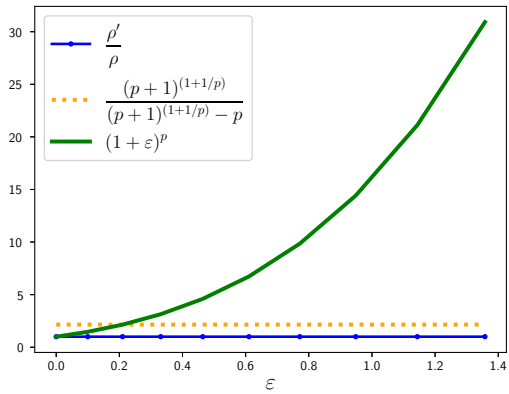}}
	\caption{Sioux Falls network with single commodity}
	\label{fig:eps4}
	\end{subfigure}
	\begin{subfigure}[b]{0.49\textwidth}\centering
	\resizebox{\textwidth}{!}{\input{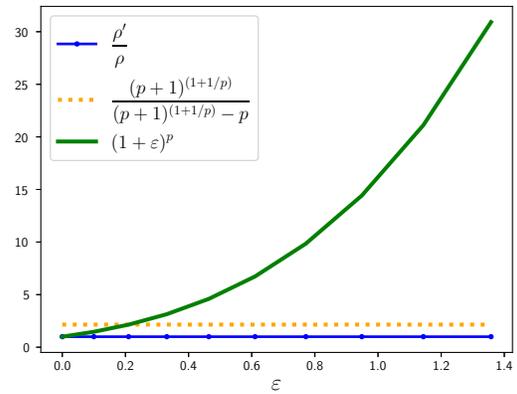}}
	\caption{Sioux Falls network with 528 commodities}
	\label{fig:eps5}
	\end{subfigure}
	\caption{Comparing PoA ratio and theoretical bounds for different $\epsilon$ values for Sioux Falls network ($p=4$).}
	\label{fig:eps6}	
\end{figure}

In general, we make a series of observations to see how hard it is to achieve the bounds in Theorem \ref{thm2}, while we do not make definitive conclusion.
First, we obtain the following two lemmas regarding the upper bound in Theorem \ref{thm2}.
Note that the upper bound is obtained only if $C'_{\mathrm{opt}} = (1+\epsilon) C_{\mathrm{opt}}$ and $C(f') = (1+\epsilon)^{p+1} C(f)$.

\begin{lemma} \label{lem:upper_opt}
Let $C_{\mathrm{opt}}$ and $C'_{\mathrm{opt}}$ be the cost of an optimal flow for instances $(G,(d_i),\ell)$ and $(G, ((1+\epsilon)d_i), \ell)$ with polynomial latency functions of degree at most $p$ with nonnegative coefficients, respectively.
Let $f^*$ and $f^{'*}$ be the optimal flows for instances $(G,(d_i),\ell)$ and $(G, ((1+\epsilon)d_i), \ell)$, respectively.
Suppose $C'_{\mathrm{opt}} = (1+\epsilon) C_{\mathrm{opt}}$.
Then, there exist edges with constant latency functions; in particular those edges $e \in E$ with $f_e^{'*} > 0$.
\end{lemma}
\begin{proof}
We let $f'^{*}$ be the optimal flow for instance $(G, ((1+\epsilon)d_i), \ell)$ and $\hat{f} = \frac{f'^*}{1+\epsilon}$.
Then, from (i), we have
\begin{align*}
	(1+\epsilon) C_{\mathrm{opt}} = C'_{\mathrm{opt}}
		& = \sum_{e\in E} \ell_e( (1+\epsilon) \hat{f}_e) (1+\epsilon) \hat{f}_e \\
    & = (1+\epsilon) \sum_{e\in E}\ell_e( (1+\epsilon)\hat{f}_e) \hat{f}_e \\
    & \geq (1+\epsilon) \sum_{e\in E}\ell_e(\hat{f}_e) \hat{f}_e \\                 
    & \geq (1+\epsilon) C_{\mathrm{opt}}.
\end{align*}
Therefore, we must have 
\[
\sum_{e\in E}\ell_e( (1+\epsilon)\hat{f}_e) \hat{f}_e = \sum_{e\in E}\ell_e(\hat{f}_e) \hat{f}_e.
\]
Consequently,
\[
\sum_{e\in E}[\ell_e( (1+\epsilon)\hat{f}_e) - \ell_e(\hat{f}_e)]\hat{f}_e = 0
\]
Since $\ell_e(\cdot)$ is monotone, we have
\begin{align*}
\hat{f}_e > 0  
& \implies \ell_e((1+\epsilon)\hat{f}_e) = \ell_e(\hat{f}_e) 
\end{align*}
for any $e\in E$.
Therefore, $\ell_e(\cdot)$ is constant on $[\hat{f}_e, (1+\epsilon)\hat{f}_e]$ for any $e\in E$ such that $\hat{f}_e > 0$. 
Since we consider polynomial latency functions, we conclude that there exist constant latency functions.
\end{proof}

\begin{lemma} \label{lem:upper_equ}
Let $f$ and $f'$ be equilibrium flows for instances $(G,(d_i),\ell)$ and $(G, ((1+\epsilon)d_i), \ell)$ with polynomial latency functions of degree at most $p$ with nonnegative coefficients, respectively.
Suppose $C(f') = (1+\epsilon)^{p+1} C(f)$. 
Then, on all edge $e\in E$ with $f_e>0$ or $f'_e>0$, the latency function $\ell_e(\cdot)$ is of order $p$.
\end{lemma}
\begin{proof}
Adopting the approach used in Theorem 3 of \citet{englert2010sensitivity}, we consider monic monomial latency functions only: $\ell_e(f_e) = f_e^{p_e}$ for each $e\in E$.
It is well known that equilibrium flows $f$ and $f'$ minimize the potential function
\[
\Phi(x) = \sum_{e\in E} \int_0^{x_e} \ell_e(u) \du
\]
on their respective feasible set. 
Therefore, we have $\Phi(f) \leq \Phi(f'/(1+\epsilon))$ and $\Phi(f') \leq \Phi((1+\epsilon)f)$, which can be written as follows:
\begin{align}
(1+\epsilon)^{p+1} \sum_{e\in E} \frac{1}{p_e+1} f_e^{p_e+1} 
&\leq
\sum_{e\in E} \frac{(1+\epsilon)^{p-p_e}}{p_e+1} f_e^{'p_e+1}, \label{A} \\
\sum_{e\in E}\frac{1}{p_e+1} f_e^{'p_e+1} 
&\leq
\sum_{e\in E} \frac{(1+\epsilon)^{p_e+1}}{p_e+1} f_e^{p_e+1}.  \label{B}
\end{align}
The condition $C(f') = (1+\epsilon)^{p+1} C(f)$ can be written as follows:
\begin{equation} \label{C}
(1+\epsilon)^{p+1} \sum_{e\in E} f_e^{p_e+1} = \sum_{e\in E} f_e^{'p_e+1}.
\end{equation}
We can express $p\cdot\eqref{A} + ((p+1)(1+\epsilon)^p-1)\cdot\eqref{B} + ((1+\epsilon)^p-1)\cdot\eqref{C}$ as follows:
\begin{equation} \label{ABC}
\sum_{e\in E} c_e f_e^{p_e+1} \leq 
\sum_{e\in E} c'_e f_e^{'p_e+1}
\end{equation}
where 
\begin{align*}
c_e  & = p \cdot \frac{(1+\epsilon)^{p+1}}{p_e+1} 
       - ((p+1)(1+\epsilon)^p-1)\cdot \frac{(1+\epsilon)^{p_e+1}}{p_e+1}
       + ((1+\epsilon)^p-1)\cdot (1+\epsilon)^{p+1} \\
c'_e & = p \cdot \frac{(1+\epsilon)^{p-p_e}}{p_e+1} 
       - ((p+1)(1+\epsilon)^p - 1) \cdot \frac{1}{p_e+1} 
       + ((1+\epsilon)^p - 1)
\end{align*}
\citet{englert2010sensitivity} show that $c_e\geq 0$ and $c'_e\leq 0$ for all $e\in E$.
Therefore, for \eqref{ABC} to hold, we must have the following condition for all $e\in E$:
if either $f_e$ or $f'_e$ is nonzero, then $c_e = c'_e = 0$.
It is easy to see that $c_e=c'_e=0$ implies $p_e=p$.
\end{proof}

Lemma \ref{lem:upper_opt} claims that the path latency function $\ell_P(\cdot)$ along any actively used path used by the optimal flow $f^{'*}$ must be a constant function.
Lemma \ref{lem:upper_equ}, on the other hand, indicates that $\ell_P(\cdot)$ along any path used by equilibrium flow either $f$ or $f'$ must be polynomials of order $p$.
Therefore, if the upper bound in Theorem \ref{thm2} is obtained, then we must have $f_e = f'_e = 0$ on any edge used by $f^{'*}$.

We make similar observations regarding the lower bound $(1+\epsilon)^{-p}$.

\begin{lemma} \label{lem:lower_opt}
Let $C_{\mathrm{opt}}$ and $C'_{\mathrm{opt}}$ be the cost of an optimal flow for instances $(G,(d_i),\ell)$ and $(G, ((1+\epsilon)d_i), \ell)$ with polynomial latency functions of degree at most $p$ with nonnegative coefficients, respectively.
Let $f^*$ and $f^{'*}$ be the optimal flows for instances $(G,(d_i),\ell)$ and $(G, ((1+\epsilon)d_i), \ell)$, respectively.
Suppose $C'_{\mathrm{opt}} = (1+\epsilon)^{p+1} C_{\mathrm{opt}}$.
Then, on all edge $e\in E$ with $f_e^{*} > 0$, the latency function $\ell_e(\cdot)$ is of order $p$.
\end{lemma}
\begin{proof}
We let $f^{*}$ be the optimal flow for instance $(G, (d_i), \ell)$ and then $(1+\epsilon) f^*$ is feasible to instance $(G, ((1+\epsilon)d_i), \ell)$.
From (i), we have
\begin{align*}
	(1+\epsilon)^{p+1} C_{\mathrm{opt}} = C'_{\mathrm{opt}} 
		& \leq \sum_{e\in E} \ell_e( (1+\epsilon) f^*_e) (1+\epsilon) f^*_e \\
		& = \sum_{e\in E} \bigg( \sum_{m=0}^p {b}_{em} ( (1+\epsilon) f^*_e )^m \bigg) (1+\epsilon) f^*_e\\		
		& \leq \sum_{e\in E} \bigg(\sum_{m=0}^p {b}_{em} (1+\epsilon)^p (f^*_e)^m \bigg) (1+\epsilon) f^*_e\\				
    & = (1+\epsilon)^{p+1} \sum_{e\in E} \ell_e(f^*_e) f^*_e \\
    & = (1+\epsilon)^{p+1} C_{\mathrm{opt}},
\end{align*}
where all inequalities must hold as equalities. 
From the second inequality, we must have
\[
\sum_{e\in E} \bigg(\sum_{m=0}^p {b}_{em} \Big[ (1+\epsilon)^p (f^*_e)^m - ( (1+\epsilon) f^*_e )^m \Big] \bigg) (1+\epsilon) f^*_e =0.
\]
For any $e\in E$, if $f^*_e >0$, then we observe that $b_{em}=0$ for $m\neq p$, which implies that $\ell_e(\cdot)$ is a monomial of order $p$.
From the first inequality, we also obtain $f^{'*}_e = (1+\epsilon) f^*_e$ for all $e\in E$.
\end{proof}

\begin{lemma} \label{lem:lower_equ}
Let $f$ and $f'$ be equilibrium flows for instances $(G,(d_i),\ell)$ and $(G, ((1+\epsilon)d_i), \ell)$ with polynomial latency functions of degree at most $p$ with nonnegative coefficients, respectively.
Suppose $C(f') = (1+\epsilon) C(f)$. 
Then, on all edge $e\in E$ with $f'_e\neq f_e$, the latency functions $\ell_e(\cdot)$ is constant.
\end{lemma}
\begin{proof}
The bound $C(f') = (1+\epsilon) C(f)$ is based on Theorem \ref{thm:dafermos}, which is derived from the monotonicity of latency functions; that is, 
\[
	\sum_{e\in E} [\ell_e(f^1_e) - \ell_e(f^2_e)] (f^1_e - f^2_e) \geq 0
\]
for any $f^1, f^2 \geq 0$.
If $C(f') = (1+\epsilon) C(f)$ holds, then, by backtracking the proof of Theorem \ref{thm:dafermos}, given in \citet{dafermos1984sensitivity}, we can easily show that 
\[
	\sum_{e\in E} [\ell_e(f') - \ell_e(f)] (f'_e - f_e) = 0.
\]
Since $\ell_e(\cdot)$ is monotone, we must have $[\ell_e(f') - \ell_e(f)] (f'_e - f_e) = 0$ for each $e \in E$.
If $f'_e \neq f_e$, then we must have $\ell_e(f'_e) = \ell_e(f_e)$, which implies that $\ell_e(\cdot)$ is a constant function.
There must exist edges with $f'_e \neq f_e$; otherwise, either $f'$ or $f$ is infeasible.
\end{proof}

Lemma \ref{lem:lower_opt} states that the latency functions are of order $p$ on all paths used by $f^*$. 
Lemma \ref{lem:lower_equ} indicates that the latency functions are constant on edges with $f'_e \neq f_e$.
Therefore, if the lower bound in Theorem \ref{thm2} is obtained, then we must have $f'_e = f_e$ on any edge used by $f^{*}$.

\section*{Acknowledgments}
This research was partially supported by the National Science Foundation under grant CMMI-1558359.

\bibliographystyle{apalike}
\bibliography{sensitivity}

\end{document}